\documentclass[journal,a4paper]{IEEEtran}
\usepackage{times}
\usepackage{cite}
\usepackage{flushend}
\usepackage{color}
\usepackage{epsf}
\usepackage{epsfig}
\usepackage{graphicx}
\usepackage{graphics}
\usepackage{epstopdf}
\usepackage[small]{caption}
\usepackage{amsmath}
\usepackage{amssymb}
\usepackage{amsthm}
\usepackage{amsxtra}
\usepackage{algorithm}
\usepackage{algorithmic}
\usepackage{enumerate}
\usepackage{multirow}
\usepackage{enumerate}
\usepackage{amssymb}
\usepackage{lipsum}
\usepackage{setspace}
\usepackage[colorlinks=false]{hyperref}
\usepackage{multirow}

\newtheorem{theorem}{\bf Theorem}

\newtheorem{definition}{\bf Definition}

\usepackage{rotating} 
\usepackage{graphicx} 

\newcommand{\Rmnum}[1]{\expandafter\@slowromancap\romannumeral #1@}

\makeatother
\begin{document}
\title{Three-Party Energy Management With Distributed Energy Resources in Smart Grid}
\author{Wayes~Tushar,~\IEEEmembership{Member,~IEEE,}
Bo~Chai,~Chau~Yuen,~\IEEEmembership{Senior Member,~IEEE,}
David B. Smith,~\IEEEmembership{Member, IEEE,} Kristin L. Wood,  Zaiyue Yang,~\IEEEmembership{Member, IEEE} and H.~Vincent~Poor,~\IEEEmembership{Fellow,~IEEE}
\thanks{W. Tushar, C. Yuen and K. L. Wood are with the Engineering Product Development at Singapore University of Technology and Design (SUTD), Dover Drive, Singapore 138682. (Email: \{wayes\_tushar, yuenchau, kristinwood\}@sutd.edu.sg).}
\thanks{B. Chai and Z. Yang are with the State Key Laboratory of Industrial Control Technology at Zhejiang University, Hangzhou, China. (Email: chaibozju@gmail.com, yangzy@zju.edu.cn).}
\thanks{David B. Smith is with the National ICT Australia (NICTA)$^\dag$, ACT 2601, Australia. D. Smith is also with the Australian National University. (Email: david.smith@nicta.com.au).}
\thanks{H. Vincent Poor is with the School of Engineering and Applied Sciences at Princeton University, Princeton, NJ, USA. (Email: poor@princeton.edu).}
\thanks{This work is supported by the Singapore University of Technology and Design (SUTD) through the Energy Innovation Research Program (EIRP) Singapore NRF2012EWT-EIRP002-045.}
\thanks{$^\dag$NICTA is funded by the Australian Government through the Department of Communications and the Australian Research Council through the ICT Centre of Excellence Program.}
}
\IEEEoverridecommandlockouts
\maketitle
\begin{abstract}
In this paper, the benefits of distributed energy resources (DERs) are considered in  an energy management scheme for a smart community consisting of a large number of residential units (RUs) and a shared facility controller (SFC). A non-cooperative Stackelberg game between RUs and the SFC is proposed in order to explore how both entities can benefit, in terms of achieved utility and minimizing total cost respectively, from their energy trading with each other and the grid.  From the properties of the game, it is shown that the maximum benefit to the SFC in terms of reduction in total cost is obtained at the unique and strategy proof Stackelberg equilibrium (SE). It is further shown that the SE is guaranteed to be reached by the SFC and RUs  by executing the proposed algorithm in a distributed fashion, where participating RUs comply with their best strategies in response to the action chosen by the SFC. In addition, a charging-discharging scheme is introduced for the SFC's storage device (SD) that can further lower the SFC's total cost if the proposed game is implemented. Numerical experiments confirm the effectiveness of the proposed scheme.
\end{abstract}
 \setcounter{page}{1}
\section{Introduction}\label{sec:introduction}
There has been an increasing interest in deploying distributed energy resources (DERs) because of their ability to reduce greenhouse gas emissions and alleviate global warming~\cite{Georgilakis-JTPS:2013}. Moreover, DERs can assist consumers in reducing their dependence on the main electricity grid as their primary source of energy~\cite{Tham-JTSMCS:2013}, and consequently can lower their cost of electricity purchase. The smart grid with enhanced communication and sensing  capabilities~\cite{Sauter-TIE:2011} offers a suitable platform for exploiting the use of DERs to assist different energy entities in effectively managing their energy with reduced dependence on the main grid.

Most literature on energy management, as we will see in the next section, has considered scenarios where users with DERs are also equipped with a storage device (SD)\cite{Justo-J-RSER:2013,Guerrero-JTIE:2013,Guerrero-JTIE:2013-2}. However, in some cases it is also likely that the users are not interested in storing energy for future use, due to the start-up and required size of a storage device. Rather, they are more concerned with consuming/trading energy as soon as it is generated, e.g., grid-integrated solar without a battery back up system~\cite{wt_battery_solar:2013}. Nevertheless, little has been done to study this kind of system. In fact, one key challenge to exploit the real benefit of using DERs in such settings is to develop appropriate system models and protocols, which are not only feasible in real world environments but, at the same time, also beneficial for the associated energy users in terms of their derived cost-benefit tradeoff. Such development further enables understanding the in-depth properties of the system and facilitates the design of suitable real time platforms for next generation power system control functions~\cite{Westermann-TIE:2010}.

To this end, we propose an energy management scheme in this paper for a smart community consisting of multiple residential units (RUs) with DERs and a shared facility controller\footnote{A dedicated authority responsible for managing shared equipments in a community.}  (SFC) using non-cooperative game theory~\cite{Ekneligoda-JTIE:2014}. To the best of our knowledge, ours is the first work to introduce the concept of a shared facility and consider a three-party energy management problem in smart grid applications. With the development of modern residential communities, shared facilities provide essential public services to the RUs, e.g., maintenance of lifts in community apartments. Hence, it is necessary to study the energy demand management of a shared facility for the benefit of the community as a whole. This is particularly necessary in the considered setting where each RU has DERs that can trade energy with both the grid and the SFC, and constitutes an important energy management problem, as we will see later, for both the SFC and RUs. Here, on the one hand, to obtain revenue, each RU would be interested in selling its energy either to the SFC or to the grid based on the prices offered by them, i.e., sell to the party with the higher price. On the other hand, the SFC wants to minimize its cost of energy purchase from the grid by making an offer to RUs such that the RUs would be more encouraged to sell their energy to the SFC instead of the grid. Thus, the SFC would need to buy less energy at a higher price from the grid. Because of the different properties and objectives of each party, the problem is more likely to handle heterogeneous customers than homogeneous ones.

Due to the heterogeneity of the distributed nodes in the system, and considering the independent decision making capabilities of the SFC and RUs, we are motivated to use a Stackelberg game~\cite{Book_Dynamicgame-Basar:1999} to design their behavior. Distinctively, we develop a distributed protocol for the SFC, which is the leader of the game, to determine the buying price from the RUs, such that its total cost of buying energy from the grid and RUs is minimized. Meanwhile, we also show how the followers, i.e., RUs without storage facilities, react in response to the buying price set by the SFC to optimize their payoffs. We extend the study by considering the case when the SFC possesses an SD\footnote{Please note that no RU possesses any SD.}, and further propose a charging-discharging scheme for the SFC that can be implemented in-line with the proposed Stackelberg game.

To this end, the main contributions of the paper are as follows: 1) A system model is proposed to facilitate energy management for the SFC and RUs in the community. Novel cost and utility models are proposed to achieve a good balance between reflecting practical requirements and providing mathematical tractability; 2) A non-cooperative Stackelberg game is proposed to capture the interaction between the SFC and the RUs. The proposed game requires limited communication between the SFC and each RU to solve the energy management problem in a decentralized fashion; 3) The properties of the game are analyzed, and the existence of a unique and strategy-proof solution is proven; 4) An algorithm is proposed that is guaranteed to reach the Stackelberg equilibrium, which can be adopted by the SFC and the RUs in a distributed fashion; and 5) A charging-discharging strategy is proposed for the SFC's storage device based on the price offered by the main grid. The introduced strategy can be implemented, along with the proposed Stackelberg game in each time slot, to further improve the SFC's benefit in terms of its total cost of energy purchase during a day.

The remainder of the paper is organized as follows. We discuss the state-of-the art of energy management research using DERs in Section~\ref{sec:state-of-the-art} followed by a description of the considered system model in Section~\ref{sec:system-model}. The energy management problem is formulated as a Stackelberg game in Section~\ref{sec:problem-formulation}, where we also analyze the properties of the game and  design a distributed algorithm. In Section~\ref{sec:with-storage}, the proposed scheme is extended to the case where the SFC possesses an SD. Numerical examples are discussed in Section~\ref{sec:case-study}, and some concluding remarks are contained in Section~\ref{sec:conclusion}.
\section{State-of-The Art}\label{sec:state-of-the-art}
Recently, there has been considerable research effort to understand the potential of DERs in smart grid~\cite{Vasques-TIE:2010}. This is mainly due to their capability in reducing greenhouse gas emissions, as well as lowering the cost of electricity~\cite{Georgilakis-JTPS:2013}. This literature can be divided into two general categories, where work such as \cite{Justo-J-RSER:2013,Guerrero-JTIE:2013,Guerrero-TIE:2011} and \cite{Guerrero-TIE:2009}  that has studied the feasibility and controls of integrating DERs in smart grid is in the first category. In \cite{Justo-J-RSER:2013}, a comprehensive literature review is provided discussing the connection and controls of AC and DC microgrid systems with DERs and energy storage systems. In \cite{Guerrero-JTIE:2013} advanced control techniques are studied, including decentralized and hierarchical controls for microgrids with distributed generation. A three-level hierarchical control process and electrical dispatching standards are presented in \cite{Guerrero-TIE:2011} with a view to integrating DERs with distributed storage systems in smart grid. A control scheme, using a droop control function for managing battery levels of domestic photovoltaic-uninterruptable power supplies (PV-UPS), is proposed in \cite{Guerrero-TIE:2009}. Other control schemes for efficient use of DERs in smart grid can be found in \cite{Georgilakis-JTPS:2013,Hill-JTSG:2012}, \cite{Liu-STSP:2014}, \cite{Naveed-Energies:2013}, \cite{ Balaguer-TIE:2011} and \cite{Liu-ISGT:2013}.

 The second category of work in this area comprises various energy management/scheduling schemes that have exploited the use of DERs in smart grid. For instance, the authors in \cite{Zhang-J_ECM:2013} study an efficient energy consumption and operation management scheme for a smart building to reduce energy expenses and gas emissions by utilizing DERs. To provide flexibility to distribution system operators, a deterministic energy management scheme is designed in \cite{Kanchev-TIE:2011} for PV generators with embedded storage. An interesting smart grid management system is explored in \cite{Cecati-TIE:2011} that uses DERs to minimize the cost of power delivery including the cost of distributed generators, the cost of power provided by the primary substation, and the cost associated with grid power losses while delivering power to the consumers. In order to minimize the operational cost of renewable integration to distributed generation systems, a forecast based optimization scheme is developed in \cite{Chakraborty-TIE:2007}. Saber \emph{et al}. propose a scheduling and controlling scheme for electric vehicles batteries in \cite{Saber-TIE:2011} so that batteries can be used and integrated with DERs for reducing emissions from electricity production. Further studies of optimization and scheduling techniques that exploit the use of DERs are available in \cite{Ramachandran-TIE:2011} and \cite{Angelis-TII:2013}.

As can be seen from the above discussion, the scope of research on the use of DERs in smart grid is not limited to power and energy research communities such as in \cite{Georgilakis-JTPS:2013,Tham-JTSMCS:2013,Justo-J-RSER:2013} and \cite{Zhang-J_ECM:2013}, but also extends to other research communities including those in smart grid~\cite{Hill-JTSG:2012,Fang-J-CST:2012}, and \emph{industrial electronics} (IE)~\cite{Guerrero-JTIE:2013,Vasques-TIE:2010,Guerrero-TIE:2011,Guerrero-TIE:2009,Balaguer-TIE:2011,Kanchev-TIE:2011,Cecati-TIE:2011,Chakraborty-TIE:2007,Saber-TIE:2011,Ramachandran-TIE:2011,Angelis-TII:2013,Yu-TIE:2011}. However, the majority of these research papers have considered the case in which all the entities with DERs also possess SDs. But this might not always be the case as we have argued in Section~\ref{sec:introduction}. In this regard, unlike the discussed literature, this paper investigates the case in which entities having DERs do not have SDs, by introducing the SFC. We use a noncooperative Stackelberg game to model the energy management scheme considering the distributed and rational nature of the nodes in the smart grid system, and thus complement the discussed previous work in the topic area. The work here has the potential to open new research opportunities for the IE and smart grid communities in terms of control of energy dispatch, size of storage devices and determination of suitable location and size of DERs that might support both the SFC and RUs to further attain different operational objectives in smart grid networks.

We stress that recent work has shown Stackelberg games to be very effective and suitable for designing energy management schemes. For example, in \cite{Maharjan-JTSG:2013}, Maharjan \emph{et al.} propose a Stackelberg game between multiple utility companies and consumers to maximize both the revenue of each utility company and the pay-off to each user. A Stackelberg game approach, using a genetic algorithm to obtain the Stackelberg solution, to maximize the profit of a electricity retailer and to minimize the payment bills of its customers, is proposed in \cite{Meng-JSpringer:2013}. A consumer-centric energy management scheme for smart grids is proposed in \cite{Tushar-TSG:2013} that prioritizes consumers' benefits by reforming their energy trading with a central power station whereby the consumers receive their socially optimal benefits at the Stackelberg equilibrium. A four-stage Stackelberg game is studied in \cite{Bu-JTETC:2013}, and analytical results are obtained via a backward induction process for electricity retailers using real-time pricing. The same authors also study the dynamics of  the smart grid in designing green wireless cellular networks in \cite{Bu-JTWC:2012} using a similar game formulation. A bi-level programming technique is used in \cite{Asimakopoulou-JTSG:2013} to design a Stackelberg game for energy management of multiple micro-grids. However, we remark that the players and their respective strategies in games significantly differ from one game to another according to the system models, design objectives and algorithms that are used. To this end, we propose a suitable system model in the next section, which can facilitate the considered energy management between the SFC, RUs and the grid through a Stackelberg game. 
\section{System Model}\label{sec:system-model}
%
\begin{figure}[t!]
\centering
\includegraphics[width=0.8\columnwidth]{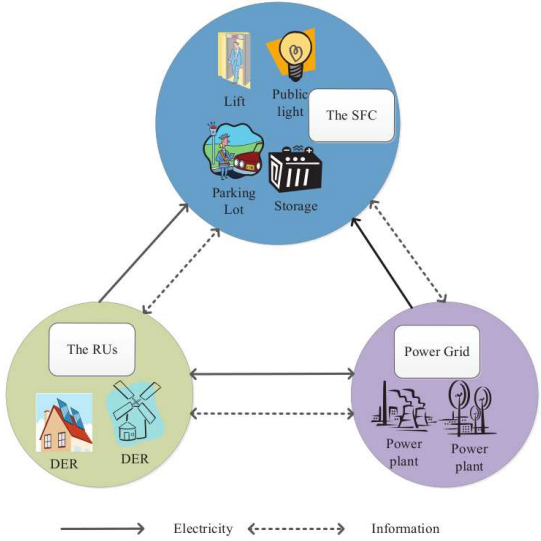}
\caption{System model for energy management in a smart community consisting of residential units, main power grid and a shared facility controller.} \label{fig:system_model}
\end{figure}
Consider a smart community consisting of a large number of RUs and an SFC. The SFC controls the electricity for equipment and machines such as lifts, water pumps, parking lot gates and lights, which are shared and used on a daily basis by the residents of the community. Here, on the one hand, the SFC does not have any electricity generation capability and, hence, needs to buy all its energy either from the main grid or from the RUs in the network.  On the other hand, each RU is assumed to have a DER without any SD that is capable of generating its own energy. An RU can be a single residential unit or a group of units, connected via an aggregator, which can act as a single entity. We assume that each RU can decide on the amount of electricity that it wants to consume, and hence the excess energy, if there is any, that it wants sell to the SFC or to the main grid for making revenue. All RUs and the SFC are assumed to be connected to one another and to the main grid by means of power and communication lines. A schematic diagram of this system is given in Fig.~\ref{fig:system_model}.

To this end, let us assume that there are $N$ RUs in the community and they belong to the set $\mathcal{N}$. Each RU $n\in\mathcal{N}$ is equipped with DERs, e.g., solar panels or wind turbines (or both), that can generate energy $E_n^\text{gen}$ at certain times during the day. We assume that each RU wants to manage its consumption $e_n$ such that it can sell the remainder of its generated energy $(E_n^\text{gen} - e_n)$ to the SFC or the grid to make revenue. Clearly, if $E_n^\text{gen} \leq E_n^\text{min}$, where $E_n^\text{min}$ is the essential load for RU $n$, the RU cannot take part in the energy management program as it cannot afford to sell any energy. Otherwise, as for the considered case, the RU  adjusts its energy consumption $e_n, \text{s.t.,} e_n\geq E_n^\text{min}$, for its own use, and thus sells the remainder $(E_n^\text{gen}-e_n)$ to the SFC or to the main grid.
\begin{figure}[t]
\centering
\includegraphics[width=\columnwidth]{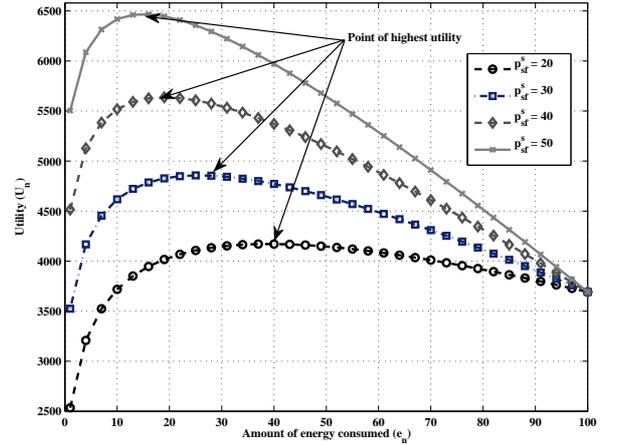}
\caption{Effect of the change of price per unit of energy paid by the SFC to each RU on the change of maximum utility that each RU receives from its energy consumption.} \label{fig:effect-p-utility}
\end{figure}

In general, the buying price $p_g^b$ set by the grid is considerably lower than its selling price $p_g^s$~\cite{McKenna-JIET:2013}. In this regard, we assume that the price per unit of energy that the SFC pays to each RU is set between the buying and selling price of the grid. Therefore, each RU can sell at a higher price, and the SFC can buy at a lower price and they trade energy with each other instead of trading with the main grid.  Under such a setting, it is reasonable to assume that all RUs would be more interested in selling $(E_n^\text{gen}-e_n)~\forall n$ to the SFC instead of the grid. Now, let us assume that the SFC sets a price $p_\text{sf}^s$ per unit of energy to pay to each RU for buying its required energy $E_\text{sf}^\text{req}$. To this end, we propose that the total utility achieved by RU $n$ from its energy consumption $e_n$ and from its trading of energy   $(E_n^\text{gen}-e_n)$ with the SFC is given by
\begin{eqnarray}
U_n = k_n\ln(1 + e_n) + p_\text{sf}^s\left(E_n^\text{gen} - e_n\right),~k_n>0.\label{eqn:utility-ru}
\end{eqnarray}
In \eqref{eqn:utility-ru}, $k_n\ln(1 + e_n)$ is the utility that the RU $n$ achieves from consuming energy $e_n$, where $k_n$ is a preference parameter~\cite{Wayes-J-TSG:2012,Samadi-C-Smartgridcomm:2010}. It is clear from \eqref{eqn:utility-ru} that an RU with higher $k_n$ would be more interested in consuming more $e_n$ to attain its maximum utility level compared to an RU with lower preference. $p_\text{sf}^s\left(E_n^\text{gen} - e_n\right)$ is the revenue that the RU receives from selling the rest of its energy to the SFC. We note that the natural logarithm $\ln(\cdot)$ has been extensively used for designing the utility~\cite{Pavlidou-JCN:2008}, and has also recently been shown to be suitable for designing the utility for power consumers~\cite{Maharjan-JTSG:2013}. We can see that \eqref{eqn:utility-ru} is a concave function and its relationship with $p_\text{sf}^s$ is shown in Fig.~\ref{fig:effect-p-utility}. As shown in Fig.~\ref{fig:effect-p-utility}, as $p_\text{sf}^s$ increases, the maximum utility of the RU shifts towards the left. That is the RU tends to sell more energy to the SFC, e.g., by scheduling the use of its flexible devices~\cite{Zhang-J_ECM:2013} to a later time, and thus becomes more interested in making further revenue.

By contrast, the SFC, having no generation capability, needs to buy all of its required energy from RUs and the main grid. In typical cases, the main grid sells energy at a higher price compared to the price from owners of renewable energy sources as for, e.g., feed-in tariff schemes~\cite{McKenna-JIET:2013}. Hence, it is reasonable to assume that the SFC would mainly be interested in buying energy from the RUs at $p_\text{sf}^s$ to meet its requirement, and procuring the rest, if there is any, from the main grid. Nonetheless, if $p_\text{sf}^s$ is too low, a RU would sell less, or no, energy to the SFC and consume more for its own purposes instead. For example, the resident may want to start washing clothes rather than schedule it at a later time, and thus use their energy instead of selling it at a very low price. Consequently, the SFC will have to buy a larger fraction of its requirement at a higher price from the main grid.  Conversely, if $p_\text{sf}^s$ is too high, e.g., close to the grid's selling price $p_g^s$, it would cost the SFC significantly. Hence, the choice of $p_\text{sf}^s$ should be within a \emph{reasonable} range to encourage the RUs to sell their energy to the SFC, but at the same time keeping the cost to the SFC at a minimum. However, if the energy from RUs is not enough to meet its requirement, the SFC needs to buy the remainder from the main grid with the price $p_g^s$. In this regard, we define a cost function to capture the total cost to the SFC for buying energy from RUs and the grid as follows:
\begin{eqnarray}
C_\text{sf} = \sum_n e_{n,\text{sf}}^s p_{\text{sf}}^s + \left(E_\text{sf}^\text{req} - \sum_n e_{n,\text{sf}}^s\right)p_g^s,
\label{eqn:cost-sfc}
\end{eqnarray}
where $e_{n,^\text{sf}}^s$ is the amount of energy that the SFC buys from the RU $n$. In \eqref{eqn:cost-sfc}, the first term captures the total cost of buying energy from the RUs. Meanwhile, the second term not only describes the cost of buying energy from the grid, but also satisfies the constraint on the demand of total required energy by the SFC, i.e., a SFC does not buy more than it requires.

Now, to decide on energy trading parameters $e_{n}$ and $p_\text{sf}^s$, on the one hand, the SFC interacts with each RU $n\in\mathcal{N}$ to minimize \eqref{eqn:cost-sfc} by choosing a suitable price to pay to each RU. On the other hand, each RU decides on the amount of energy it wants to sell to the SFC by controlling its energy consumption $e_n$ so as to maximize \eqref{eqn:utility-ru}. To this end, we design the interaction and energy trading behavior of each energy entity in the next section.
\section{Energy Management Between RUs and the SFC via Game Theory}\label{sec:problem-formulation}
\subsection{Objective of the RU}\label{sec:objective-ru}
First we note that \eqref{eqn:utility-ru} and \eqref{eqn:cost-sfc} are coupled through $E_n^\text{gen}, p_\text{sf}^s$ and $e_n$. Since the RUs do not have any storage capacity, each RU would desire to sell all its excess energy,
\begin{eqnarray}
e_{n,\text{sf}}^s = E_n^\text{gen} - e_n.\label{eqn:sell-with-generation}
\end{eqnarray}
at a suitable price $p_\text{sf}^s$ to the SFC after adjusting for their consumption $e_n$. To that end, the objective of each RU $n$ can be defined as
\begin{eqnarray}
\max_{e_n}~ U_n,\nonumber\\\text{s.t.,}~e_n\geq E_n^\text{min}.\label{eqn:obj-ru}
\end{eqnarray}
Now from \eqref{eqn:utility-ru} and \eqref{eqn:obj-ru}, the first-order-differential condition for maximum utility is
\begin{eqnarray}
\frac{k_n}{1 + e_n} - p_{\text{sf}}^s = 0,
\end{eqnarray}
and hence
\begin{eqnarray}
e_n = \frac{k_n}{p_\text{sf}^s} - 1,\label{eqn:relation-en-pn}
\end{eqnarray}
which clearly relates the decision making process of each RU to the price set by the SFC. Here, $k_n$ should be sufficiently large such that \eqref{eqn:relation-en-pn} always possesses a positive value for all resulting values of $e_n$ and $p_\text{sf}^s$, and $e_n$ is at least as large as its essential load. From~\eqref{eqn:relation-en-pn}, the amount of energy $e_n$ chosen to be consumed by each RU is inversely proportional to the price per unit of energy paid by the SFC to the RU. As a result, for a higher $p_\text{sf}^s$, the RU $n$ would be more inclined to sell to the SFC by reducing its consumption and vice-versa.
\subsection{Objective of the SFC}\label{sec:objective-sfc}
In contrast, the objective of the SFC is to minimize its total cost of buying energy. Since the SFC does not have any control over the pricing of the grid, it can only set its own buying price $p_\text{sf}^s$ to minimize \eqref{eqn:cost-sfc}. Hence, the objective of the SFC is
\begin{eqnarray}
\min_{p_\text{sf}^s}C_\text{sf}.
\label{eqn:obj-sfc}
\end{eqnarray}
Now, from the first order optimality condition of the SFC's objective function~\eqref{eqn:cost-sfc},
\begin{eqnarray}
\frac{\delta C_\text{sf}}{\delta p_\text{sf}^s} = 0.\label{eqn:cond-sfc-1}
\end{eqnarray}
By replacing $e_{n,\text{sf}}^s$ in \eqref{eqn:cost-sfc} with $(E_n^\text{gen} - e_n)$, and considering the relationship between $e_n$ and $p_\text{sf}^s$ from \eqref{eqn:relation-en-pn}, we obtain
\begin{eqnarray}
\frac{\delta}{\delta p_\text{sf}^s}\Bigg(\sum_n(E_n^\text{gen} &-& k_n + p_\text{sf}^s) + E_\text{sf}^\text{req}p_g^s\nonumber\\ &-& p_g^s\sum_n\left(E_n^\text{gen} - \frac{k_n}{p_\text{sf}^s} + 1\right)\Bigg) = 0.\label{eqn:cond-sfc-2}
\end{eqnarray}
And from \eqref{eqn:cond-sfc-2}, we derive
\begin{eqnarray}
p_\text{sf}^s = \sqrt{\frac{p_g^s\sum_n k_n}{N + \sum_n E_n^\text{gen}}},\label{eqn:price-relation}
\end{eqnarray}
which emphasizes that the optimal price set by the SFC is influenced by the total number of RUs that wish to sell their energy and the generation of their DERs during the considered time. It is also established from~\eqref{eqn:price-relation} that $p_\text{sf}^s$ is affected by the grid's price, which consequently influences the SFC to change its per unit price for the RUs. However, as discussed in Section~\ref{sec:system-model}, to encourage the RUs to always sell their excess energy to the SFC we propose that the choice of price by the SFC is
\begin{eqnarray}
p_\text{sf}^s = \begin{cases}
 \sqrt{\frac{p_g^s\sum_n k_n}{N + \sum_n E_n^\text{gen}}}, & \text{if $p_\text{sf}^s>p_g^b$}\\
 p_g^b + \alpha, & \text{otherwise}.
\end{cases}
\label{eqn:price-relation-2}
\end{eqnarray}
Here, $\alpha>0$ is a small value to keep $p_\text{sf}^s$ higher than $p_g^b$.

It is obvious from \eqref{eqn:price-relation-2} that the SFC can optimize its price in a centralized fashion to minimize its total cost of purchasing energy from the RUs and the grid, if it has full access to the private information of each RU $n$, such as $E_n^\text{gen}$ and $k_n$ . However, in reality, it might not be possible for the SFC to access this information in order to protect the users' privacy, and hence a distributed mechanism is necessary to determine the parameters $p_\text{sf}^s$ and $e_n,~\forall n$. To that end, we propose a scheme based on a non-cooperative Stackelberg game in the following section.
\subsection{Non-cooperative Stackelberg game}\label{sec:stackelberg-game}
A Stackelberg game formally studies the multi-level decision making processes of a number of \emph{independent} decision makers (i.e., followers) in response to the decision taken by the \emph{leading} player (leader) of the game~\cite{Book_Dynamicgame-Basar:1999}. In this section, we formulate a non-cooperative Stackelberg game, where the SFC is the leader, and RUs are the followers, to capture the interaction between the SFC and the RUs. The game is formally defined by its strategic form as
\begin{eqnarray}
\Gamma = \{(\mathcal{N}\cup\{\text{SFC}\}), \{\mathbf{E}_n\}_{n\in\mathcal{N}}, \{U_n\}_{n\in\mathcal{N}}, p_\text{sf}^s, C_\text{sf}\},
\label{eqn:definition-game}
\end{eqnarray}
which consists of the following components:
\begin{enumerate}[i)]
\item The RUs in set $\mathcal{N}$ act as followers and choose their strategies in response to the price set by the SFC, i.e., the leader of the game.
\item $\mathbf{E}_n$ is the set of strategies of each RU $n\in\mathcal{N}$ from which it selects its strategy, i.e., the amount of energy $e_n\in\mathbf{E}_n$ to be consumed during the game.
\item $U_n$ is the \emph{utility function} of each RU $n$ as explained in \eqref{eqn:utility-ru} that captures the RU's benefit from consuming energy $e_n$ and selling energy $(E_n^\text{gen} - e_n)$ to the SFC.
\item $p_\text{sf}^s$ is the price set by the SFC to buy energy from the RUs.
\item The \emph{cost function} $C_\text{sf}$ of the SFC captures the total cost incurred by the SFC for trading energy with RUs and the main grid.
\end{enumerate}

As discussed previously, the objectives of each RU and the SFC are to maximize the utility in \eqref{eqn:utility-ru} and to minimize the cost in \eqref{eqn:cost-sfc} respectively by their chosen strategies. For this purpose, one suitable solution for the proposed game is the Stackelberg equilibrium (SE) at which the leader obtains its optimal price given the followers' best responses. At this equilibrium, neither the leader nor any follower can benefit, in terms of total cost and utility respectively, by \emph{unilaterally} changing their strategy.
\begin{definition}
Consider the game $\Gamma$ defined in \eqref{eqn:definition-game}, where $U_n$ and $C_\text{\emph{sf}}$ are determined by \eqref{eqn:utility-ru} and \eqref{eqn:cost-sfc} respectively. A set of strategies $\left(\mathbf{e}^{*}, p_\text{sf}^{s*}\right)$ constitutes an SE of this game, if and only if it satisfies the following set of inequalities:
\begin{eqnarray}
U_n(\mathbf{e}^{*}, p_\text{\emph{sf}}^{s*})\geq U_n(e_n,\mathbf{e}_{-n}^{*}, p_\text{\emph{sf}}^{s*}),~\forall n\in\mathcal{N},~\forall e_n\in\mathbf{E}_n,\label{eqn:definition-1}
\end{eqnarray}
and
\begin{eqnarray}
C_\text{\emph{sf}}(\mathbf{e}^{*}, p_\text{\emph{sf}}^{s*})\leq C_\text{\emph{sf}}(\mathbf{e}^{*}, p_\text{\emph{sf}}^s),\label{eqn:definition-2}
\end{eqnarray}
where $\mathbf{e}_{-n}^{*} = \left[e_{1}^{*}, e_{2}^{*},\hdots, e_{n-1}^{*}, e_{n+1}^{*},\hdots, e_{N}^{*}\right]$ and $\mathbf{e}^{*} = \left[e_n^*, \mathbf{e}_{-n}^*\right]$.
\label{definition:1}
\end{definition}
Therefore, when all the players in $\left(\mathcal{N}\cup\{\text{SFC}\}\right)$ are at an SE, the SFC cannot reduce its cost by reducing its price from the SE price $p_\text{sf}^{s*}$, and similarly, no RU $n$ can improve its utility by choosing a different energy to $e_n^*$ for consumption.
\subsection{Existence and Uniqueness of SE}\label{sec:existence-SE}
In non-cooperative games, an equilibrium in pure strategies is not always guaranteed to exist~\cite{Book_Dynamicgame-Basar:1999}. Therefore, we need to investigate as to whether there exists an SE in the proposed Stackelberg game.
\begin{theorem}
A unique SE always exists in the proposed Stackelberg game $\Gamma$ between the SFC and RUs in the set $\mathcal{N}$.
\label{theorem:1}
\end{theorem}
\begin{proof}
First, we note that the utility function $U_n$ in \eqref{eqn:utility-ru} is strictly concave with respect to $e_n~\forall n\in\mathcal{N}$, i.e., $\frac{\delta^2 U_n}{\delta e_n^2}<0$, and hence for any price $p_\text{sf}^s>0$, each RU $n$ will have a unique $e_n$, chosen from a bounded range $\left[E^\text{min}_n, E_n^\text{gen}\right]$, that maximizes $U_n$. We also note that the game $\Gamma$ reaches the SE when all the players in the game, including each participating RU and the SFC, have their optimized payoff and cost respectively, considering the strategies chosen by all players in the game. Thereby, it is evident that the proposed game $\Gamma$ reaches an SE as soon as the SFC is able to find an optimized price $p_\text{sf}^{s*}$ while the RUs choose their unique energy vector $\mathbf{e}^{*}$.

Now from \eqref{eqn:cond-sfc-2}, given the choices of energy by each RU $n$ in the network, the second derivative of $C_\text{sf}$ is
 \begin{eqnarray}
\frac{\delta^2 C_\text{sf}}{\delta p_\text{sf}^{s^2}} = \frac{2\sum_n k_n}{(p_\text{sf}^s)^3}>0,
 \end{eqnarray}
 and therefore, $C_\text{sf}$ is strictly convex with respect to $p_\text{sf}^s$. Hence, the SFC would be able to find an optimal unique per-unit price $p_\text{sf}^{s*}$ for buying its energy from the RUs  based on their strategies. Therefore, there exists a unique SE in the proposed game, and thus Theorem~\ref{theorem:1} is proved.
\end{proof}
In the next section, we propose an algorithm that all the RUs and the SFC can implement in a distributed fashion to reach the unique SE. We note that it is also possible to solve the energy management problem in a centralized fashion if we have global information such as $E_n^\text{gen}$ and $k_n$  available at the SFC. However, in order to protect the privacy of each RU and also to reduce the demand on communications bandwidth, a distributed algorithm is desired where the optimization can be performed by each RU and the SFC without the need for any private information to be available at the SFC.
\subsection{Distributed Algorithm}\label{sec:algorithm}
\begin{algorithm}[h]
\caption{Algorithm to reach the SE}
\label{alg:1}
\begin{algorithmic}[1]
\small
\STATE Initialization: $p_\text{sf}^{s*}=0$ $C_\text{sf}^*=p_g^s*E_\text{sf}^\text{req}$
\FOR {Buying pricing $p_\text{sf}^s$  from $p_g^b$ to $p_g^s$ }
   \FOR {Each RU $n \in \mathcal{N}$}
        \STATE RU $n$ adjusts its energy consumption $e_n$ according to
        \begin{equation}\label{eqn:alg-1}
           e_n^* = {\rm{arg}}{\kern 1pt} {\kern 1pt} {\kern 1pt} \mathop {\max }\limits_{0 \le {e_n} \le E_n^\text{gen}} {\kern 1pt} {\kern 1pt} {\kern 1pt} [{k_n}\ln (1 + {e_n}) + p_\text{sf}^s(E_n^\text{gen} - {e_n})].
        \end{equation} \\
   \ENDFOR
    \STATE The SFC computes the cost according to
        \begin{equation}\label{eqn:alg-2}
    {C_\text{sf}} = p_\text{sf}^s\sum\limits_{n \in \mathcal{N}} {\left(E_n^\text{gen} - {e_n}\right)}  + p_g^s\left(E_\text{sf}^\text{req} -
     \sum\limits_{n \in \mathcal{N}} {\left(E_n^\text{gen} - {e_n}\right)} \right).
        \end{equation} \\
     \IF {$C_\text{sf} \le C_\text{sf}^*$}
         \STATE The SFC records the optimal price and minimum cost
         \begin{equation}\label{eqn:alg-3}
           p_\text{sf}^{s*}=p_\text{sf}^{s}, C_\text{sf}^*=C_\text{sf}
         \end{equation}
     \ENDIF
\ENDFOR\\
\textbf{The SE $(\mathbf{e}^*, p_\text{sf}^{s*})$ is achieved.}
\end{algorithmic}
\end{algorithm}
In order to attain the SE, the SFC needs to communicate with each RU. We propose an algorithm that all the RUs and the SFC can implement in a distributed fashion to iteratively reach the unique SE of the proposed game. In each iteration, firstly the RU $n$ chooses its best energy consumption amount $e_n$ in response to the price set by the SFC, calculating $e_{n,\text{sf}}^s = (E_n^\text{gen} - e_n)$, and sending this information to the SFC. Secondly, having the information about the choices of energy $\mathbf{e}_\text{sf}^s = [e_{1,\text{sf}}^s, e_{2,\text{sf}}^s,\hdots, e_{N,\text{sf}}^s]$ by all RUs, the SFC decides on its best price that minimizes its total cost according to \eqref{eqn:cost-sfc}. The interaction continues until the conditions in \eqref{eqn:definition-1} and \eqref{eqn:definition-2} are satisfied, and therefore the Stackelberg game reaches the SE. Details are given in Algorithm~\ref{alg:1}.

In the proposed algorithm, the conflict between the RUs' choices of strategies stem from their impact on the choice of $p_\text{sf}^s$ by the SFC. Due to the strict convexity of $C_\text{sf}$, the choice of $p_\text{sf}^{s*}>0$ lowers the cost of the SFC to the minimum. Now, as the algorithm is designed, in response to the $p_\text{sf}^{s*}$, each RU $n$ chooses its strategy $e_n$ from the bounded range $\left[E_n^\text{min}, E_n^\text{gen}\right]$ to maximize its concave utility function $U_n$. Hence, due to the bounded strategy set and the continuity of the utility function $U_n$ with respect to $e_n$, each RU $n$ also reaches a fixed point at which its utility is maximized for the given price $p_\text{sf}^{s*}$~\cite{Maharjan-JTSG:2013}. As a consequence, the proposed algorithm is always guaranteed to converge to the unique SE of the game.

\subsubsection{Strategy-Proof Property}\label{sec:strategy-proof}
Since, each RU plays its best response in Algorithm~\ref{alg:1}, it is important to investigate whether RUs can choose a different strategy or cheat other players in $\Gamma$ once they reach the SE. In this regard, we would like see whether it is possible for an RU to change the amount of energy that it offers to the SFC, i.e., $e_{n,\text{sf}}^s = (E_n^\text{gen}-e_n)$ by changing the energy consumption $e_n$ while using Algorithm~\ref{alg:1}.
\begin{theorem}
It is not possible for any RU $n\in\mathcal{N}$ to be untruthful about its strategy, i.e., sell more or less than what it promises to give the SFC, when all other players including the SFC and the RUs in $\mathcal{N}/\{n\}$ are adopting Algorithm~\ref{alg:1}.
\label{theorem:2}
\end{theorem}
\begin{proof}
To prove Theorem~\ref{theorem:2}, first we consider that $p_\text{sf}^{s*}$ and $\mathbf{e}^{*} = \left[e_1^*,\hdots,e_n^*,\hdots, e_N^*\right]$ are the SE solutions of the proposed game obtained via Algorithm~\ref{alg:1}. Let us assume that RU $n$ is untruthful, and chooses $e_n^{'}$ instead of $e_n^{*}$ to consume after reaching the SE. Therefore, from \eqref{eqn:relation-en-pn},
\begin{eqnarray}
e_n^{'} = \frac{k_n}{p_\text{sf}^{s*}} - 1,\label{eqn:strategy-proof}
\end{eqnarray}
which is impossible. This is due to the fact that, as the scheme is formulated, $p_\text{sf}^{s*}$ results from Algorithm~\ref{alg:1} only if all the RUs in $\mathcal{N}$ consume the SE amount of energy $e_n^{*}~\forall n$. In this regard, \eqref{eqn:strategy-proof} is only true if $e_n^{'} = e_n^{*}$, which successively proves the \emph{strategy proof} property of the proposed algorithm.
\end{proof}
\section{Energy Management with Storage}\label{sec:with-storage}
We note that the proposed scheme in Section III determines the best price for the SFC to minimize its total cost of energy purchase at any given time. The scheme also benefits the RUs in terms of their energy consumption and trading with the SFC. However, DERs do not provide a stable supply. Sometimes there could be an abundant supply of energy whereas at other times there could be a scarcity. In other words, sometimes the SFC might need to buy less energy from the grid whereas at other times it might need to buy a larger amount. Following from these characteristics, we propose a storage scheme for the SFC in this section that can further reduce its total cost, if the scheme is implemented in conjunction with the Stackelberg game.

We assume that the SFC is equipped with an SD, and the charging and discharging of the SD at different times of the day is carried out based on the time of use (ToU) price~\cite{Wang-JTSG:2013} announced by the grid. The intuition behind considering a ToU price as the baseline for the SD's charging-discharging can be explained as follows: 1) Since the SFC does not know the private information of RUs, such as their energy generation and preferences, it cannot determine how much energy it can buy from them (and the associated cost) ahead of time. Hence, by allowing a ToU price to decide its charging and discharging, the SD can leverage the flexibility of the SFC in trading energy with the grid in the event of energy scarcity at the RU at any time of the day. And 2) It is reasonable and practical to assume that the grid's ToU price is announced ahead of time \cite{Wang-JTSG:2013}. Hence, it would be more practical for the SFC to decide on the charging and discharging pattern of its SD based on a known price that gives a mathematically tractable solution.

To this end, we consider that the total time of energy management during a day is divided into $T$ time slots where each time slot $t$ has a duration of one hour~\cite{Jin-J-TVT:2013}. At $t$, the total requirement of energy $E_\text{sf}^\text{req}(t)$ by the SFC has two components:
\begin{eqnarray}
E_\text{sf}^\text{req}(t) = E_\text{sf}^\text{eqp}(t) + e_\text{sf}^\text{SD}(t),\label{eqn:requirement-withSD}
\end{eqnarray}
where $E_\text{sf}^\text{eqp}(t)$ is the amount of energy exclusively required to run equipment of the shared facility, and $e_\text{sf}^\text{SD}(t)$ is the energy charged-to/discharged-from the SFC's storage at $t$.
 \begin{figure}[t]
\centering
\includegraphics[width=\columnwidth]{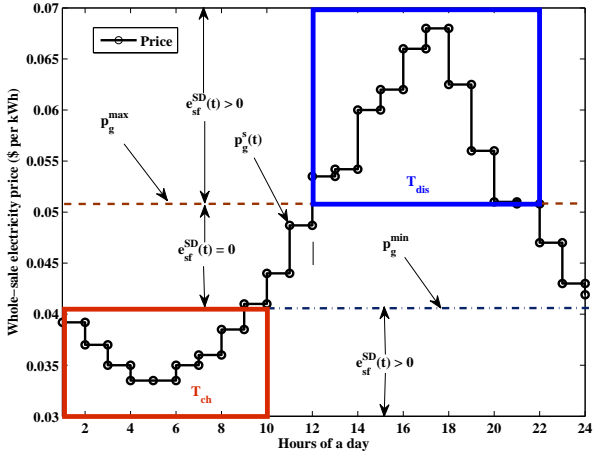}
\caption{Choice of charging and discharging duration of the SFC's SD based on the price announced by the grid.} \label{fig:charging-discharging-duration}
\end{figure}
It is assumed that the announced price per time slot is available from the main grid ahead of time~\cite{Wang-JTSG:2013}. We suppose that the SFC selects two price levels $p_g^\text{min}$ and $p_g^\text{max}$ as the minimum and maximum price thresholds from the announced price list. The SFC charges its battery at time $t$ if $p_g^s(t)<p_g^\text{min}$, and discharges the battery if $p_g^s(t)>p_g^\text{max}$. The duration of time at which these two conditions are satisfied are characterized as the charging duration $T_\text{chg}$ and the discharging duration $T_\text{dis}$ respectively. The choice of $T_\text{chg}$ and  $T_\text{dis}$ based on $p_g^s(t)~\forall t\in T, p_g^\text{min}$ and $p_g^\text{max}$ are shown graphically in Fig.~\ref{fig:charging-discharging-duration}. To that end, the charging and discharging process of the SFC's SD can be implemented as follows:

\textbf{Charging of the SD:}
\begin{enumerate}
\item The SFC, with an initial state of charge (SOC) at the SD $Q_\text{ini}$, determines $T_\text{chg}$ according to the price announced by the grid and selected $p_g^\text{min}$.
\item For each $t\in T_\text{chg}$, the SFC derives $p_g^\text{min} - p_g^s(t)$, and checks the total $\sum_{t\in T_\text{chg}}(p_g^\text{min} - p_g^s(t))$ for the full charging duration.
\item The SFC sets a target SOC at the end of the charging duration, i.e., $Q_\text{tar}^\text{ch}$, and charges its SD at each $t$ based on the proportion of price difference between time slots,  $Q_\text{tar}^\text{ch}$, and $Q_\text{ini}$  through
\begin{eqnarray}
e_\text{sf}^\text{SD}(t) = \frac{\left(p_g^\text{min}-p_g^s(t)\right)\left(Q_\text{tar}^\text{ch} - Q_\text{ini}\right)\sigma}{\sum_{t\in T_\text{chg}}\left(p_g^\text{min}-p_g^s(t)\right)},~\forall t\in T_\text{chg},\label{eqn:charging-sd}
\end{eqnarray}
where $\sigma$ is the efficiency of SFC's SD.
\end{enumerate}
We stress that the SFC cannot charge its SD at a rate more than its maximum allowed charging rate~\cite{Jin-J-TVT:2013}. Hence, \eqref{eqn:charging-sd} can be modified as
\begin{eqnarray}
e_\text{sf}^\text{SD}(t) = \min\left(\frac{\left(p_g^\text{min}-p_g^s(t)\right)\left(Q_\text{tar}^\text{ch} - Q_\text{ini}\right)\sigma}{\sum_{t\in T_\text{chg}}\left(p_g^\text{min}-p_g^s(t)\right)}, e_\text{sf}^\text{max}\right),\label{eqn:charge-sd}
\end{eqnarray}
where $e_\text{sf}^\text{max}$ is the maximum charging/discharging rate of the SFC's SD.

\textbf{Discharging of the SD:} The SFC discharges its SD at each time slot $t\in T_\text{dis}$ following a similar process to that described in the previous paragraph. Therefore, at each $t\in T_\text{dis}$: 1) the SFC derives $(p_g^s(t) - p_g^\text{max})$ and determines the overall $\sum_{t\in T_\text{dis}}\left(p_g^s(t) - p_g^\text{max}\right)$ for the whole duration of $T_\text{dis}$; and then 2) based on the proportion of price difference between discharging time slots, achieved SOC $Q_\text{tar}^\text{ch}$ during $T_\text{chg}$, and the target SOC $Q_\text{tar}^\text{dis}$ at the end of discharging period $T_\text{dis}$, the SFC discharges its SD using
\begin{eqnarray}
e_\text{sf}^\text{SD}(t) = - \min\Bigg(\frac{(p_g^s(t) - p_g^\text{max})(Q_\text{tar}^\text{ch}-Q_\text{tar}^\text{dis})\sigma}{\sum_{t\in T_\text{dis}}(p_g^s(t) - p_g^\text{max})},\nonumber\\ e_\text{sf}^\text{max}, E_\text{sf}^\text{eqp}(t)\Bigg),\label{eqn:discharge-sd}
\end{eqnarray}
for all $t\in T_\text{dis}$. As shown in \eqref{eqn:discharge-sd}, during discharge, the SFC cannot drain its SD by more than what is required by equipment as it would result in a negative requirement in \eqref{eqn:requirement-withSD}. The negative sign in \eqref{eqn:discharge-sd} emphasizes that the SD is discharging during $T_\text{dis}$.

By adopting \eqref{eqn:charge-sd} and \eqref{eqn:discharge-sd} the SFC is enabled to charge its SD during lower price periods and discharge it during higher price periods, which consequently reduces the cost of energy trading of the SFC. We note that a similar idea has been used previously to reduce the energy consumption cost of different energy entities by using batteries~\cite{Fang-J-CST:2012}. However, in this work the inclusion of a Stackelberg game with this charging-discharging scheme in each time slot makes the RUs with DER part of the system, and thus significantly further reduces the costs to the SFC, as will be shown via numerical experiments in the next section. The choice of two thresholds provides the SFC with the flexibility to choose different ranges of charging, discharging and idle durations. For example,  if $p_g^\text{max} = p_g^\text{min}$, the threshold of the proposed scheme would merge with the choice of threshold proposed in \cite{Wang-JTSG:2013}. However, the decision making mechanism of the charging and discharging amount at each time as proposed in this paper is completely different from that in \cite{Wang-JTSG:2013}.
\section{Case Study}\label{sec:case-study}
%
 \begin{figure}[t]
\centering
\includegraphics[width=\columnwidth]{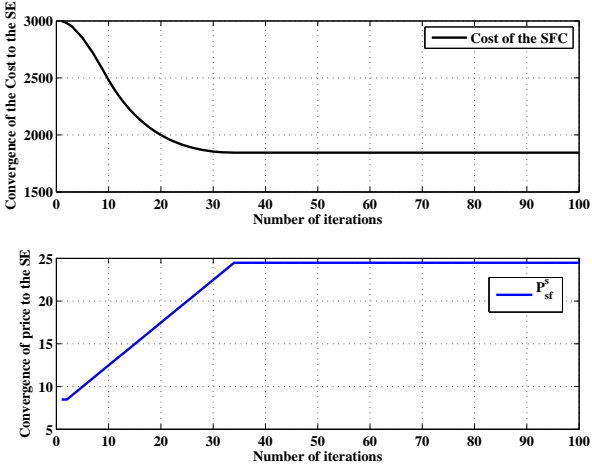}
\caption{Convergence of the proposed scheme to the SE.} \label{fig:convergence}
\end{figure}
For numerical case studies, we consider a number of RUs in the smart community that are interested in selling their energy to the SFC. Typical energy generation of each RU $n$ from its DERs is assumed to be 10 kWh~\cite{NREL_wind_generation:2009}, and is considered to be the same for all RUs in the network. The required energy by the SFC is  assumed to be 50 kWh during the considered time. As shown in \eqref{eqn:relation-en-pn}, the minimum requirement of each RU $n$ depends on its preference parameter $k_n$, which is chosen separately for different RUs, and is considered to be sufficiently large\footnote{For this case study, each $k_n$ is generated as a uniformly distributed random variable from the range $\left[90, 150\right]$.} such that \eqref{eqn:relation-en-pn} does not possess any negative values and $e_n$ is at least equal to $E_n^\text{min}$. The grid's per-unit sale price is assumed to be 60 cents~\cite{Jin-J-TVT:2013}, whereby the SFC sets its initial price to be $8.45$ cents per kWh\footnote{Which is the buy back price of the grid~\cite{Tushar-TSG:2013}.} to pay to each RU. It is very important to highlight that all parameter values are particular to this study and may vary according to the need of the SFC, power generation of the grid, weather conditions of the day, time of the day/year, and the country.

In Fig.~\ref{fig:convergence}, the convergence of the SFC's total cost to the SE by following Algorithm~\ref{alg:1} is shown for a network of five RUs. Here we see that although the SFC wants to minimize its total cost, it cannot manage to do so with its initial choice of price for payment to the RUs. In fact, through interaction with each RU of the network, the SFC eventually increases its price in each iteration to encourage the RUs to sell more, and consequently the cost continuously reduces. As can be seen from Fig.~\ref{fig:convergence}, the SFC's choice of equilibrium price, and consequently the minimum total cost, reaches its SE after the $34^\text{th}$ iteration.
 \begin{figure}[b]
\centering
\includegraphics[width=\columnwidth]{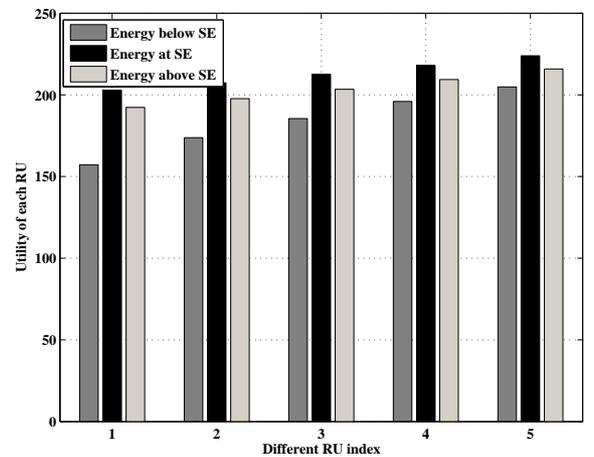}
\caption{Utility achieved by each RU at the SE.} \label{fig:utility-each-RU}
\end{figure}

As the SFC's total cost of energy purchase reaches its SE, the RUs in the network also reach their best utilities by playing their best strategies in response to the price offered by the SFC. We show the utility achieved by each RU at the SE in Fig.~\ref{fig:utility-each-RU}. As discussed in Definition~\ref{definition:1}, it is shown in Fig.~\ref{fig:utility-each-RU} that any deviation from the choice of energy consumption at the SE assigns a lower utility to the RU. In Fig.~\ref{fig:utility-each-RU} we compare both of the cases: 1) choice of energy more than the SE amount and 2) choice of energy lower than the SE amount, with the SE energy choice by each RU. It shows that only the SE assigns the maximum utility to each RU, and thus establishes a stable solution of the game.
\begin{table}[t]
\centering
\caption{Effect of the number of RUs on the total cost (in dollar) incurred by the SFC ($E^\text{req}_\text{sf} = 150$ kWh, $p_g^s = 70$ cents/kWh).}
\begin{tabular}{|c|c|c|c|c|c|}
\hline
Number of RU & 5 & 10 & 15 & 20 & 25\\
\hline
Cost (Baseline) & 105 & 105 & 105 & 105 & 105\\
\hline
Cost (Proposed) & 84.23 & 64.38 & 44.20 & 23.79 & 2.78 \\
\hline
$\%$ Reduction & 19.78 & 38.68 & 57.89 & 77.34 & 97.34\\
\hline
\end{tabular}
\label{tab:cost-vs-rus}
\end{table}
%
\begin{table}[t]
\centering
\caption{Effect of change of SFC's required energy on its total cost in dollars ($N = 10, p_g^s = 70$ cents/kWh).}
\begin{tabular}{|c|c|c|c|c|c|}
\hline
$E^\text{req}_\text{sf}$ & 60 & 70 & 80 & 90 & 100\\
\hline
Cost (Baseline) & 42.0 & 49.0 & 56.0 & 63.0 & 70.0\\
\hline
Cost (Proposed) & 1.384 & 8.384 & 15.38 & 22.38 & 29.38 \\
\hline
$\%$ Reduction & 96.70 & 82.89 & 72.53 & 64.47 & 58.02\\
\hline
\end{tabular}
\label{tab:cost-vs-req}
\end{table}

In Tables~\ref{tab:cost-vs-rus} and \ref{tab:cost-vs-req}, we investigate how the proposed scheme captures the change in total cost to the SFC  as different parameters, such as the number of RUs and the SFC's energy requirement,  change in the system. We compare the results with a baseline approach that does not have any DER facility, i.e., the SFC depends on the grid for all its energy. First, in Table~\ref{tab:cost-vs-rus}, the cost to the SFC is shown to gradually decrease for the proposed case as the total number of RUs increases in the network. This is due to the fact that as the number of RUs increases in the system, the SFC can buy more energy at a cheaper rate from more RUs, and consequently becomes less dependent on the grid's more expensive energy. Hence, the cost reduces eventually. However, due to the absence of any DERs, the cost to the SFC does not change with number of RUs in the network in the baseline approach, and the cost is shown to be significantly higher than for the proposed scheme. From Table.~\ref{tab:cost-vs-rus}, on average the cost reduction is $58.2\%$ for the proposed case compared to the baseline approach, with the considered parameter values, as the number of RUs varies in the system.

Whereas the cost to the SFC decreases with an increase in RUs in the system, we observe the opposite effect on cost while the SFC's energy requirement increases. As shown in Table~\ref{tab:cost-vs-req}, the cost to the SFC increases for both the proposed and baseline approaches as the energy required by the SFC increases. In fact, it is trivial to observe that needing more energy leads the SFC to spend more on buying energy, which consequently increases the cost. Nonetheless, the proposed scheme needs to spend less to buy the same amount of energy due to the presence of DERs of the RUs, and thus noticeably benefits from its energy trading, in terms of total cost, when compared to the baseline scheme. From Table~\ref{tab:cost-vs-req}, the SFC's average cost is $74.9\%$ lower than that of the baseline approach for the considered changes in SFC's energy requirements.

\begin{figure}[t]
\centering
\includegraphics[width=\columnwidth]{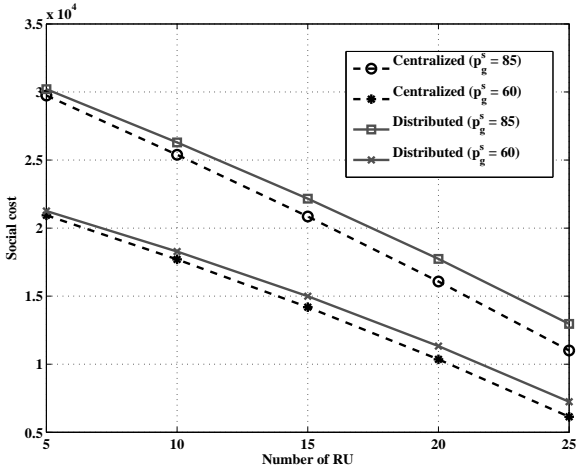}
\caption{Comparison of social cost obtained by the proposed distributed scheme with respect to the centralized scheme as the number of RUs varies in the network.} \label{fig:CentralVsPropose}
\end{figure}
As we have discussed above, it is also possible to optimally manage energy between RUs and the SFC via a centralized control system to minimize the social cost\footnote{Which is the difference between the total cost incurred by the SFC and total utility achieved by all RUs in the system.} if private information such as $k_n$ and $E_n^\text{gen}~\forall n$ is available to the controller. In this regard, we observe the performance in terms of social cost for both the centralized and proposed distributed  schemes for two different price schemes in Fig.~\ref{fig:CentralVsPropose}. As can be seen from the figure, the social cost attained by adopting the distributed scheme  is close to the optimal centralized scheme at the SE of the game in both the cases. However, the centralized scheme has access to the private information of each RU. Hence, the controller can optimally manage the energy, and as a result there is better performance in terms of reducing the SFC's cost compared to the proposed scheme. According to Fig.~\ref{fig:CentralVsPropose}, as the number of RUs increases in the network from $5$ to $25$, the average social cost for the proposed distributed scheme is only $7.07\%$ and $6.75\%$ higher than for the centralized scheme for $p_g^s = 85$ and $60$ cents/kWh respectively. This is a very promising result considering that the system is distributed.

\begin{figure}[t]
\centering
\includegraphics[width=\columnwidth]{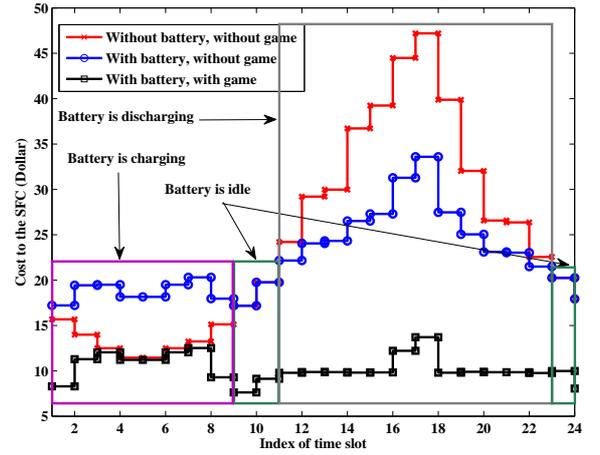}
\caption{Comparison of the cost incurred by the SFC at different times of the day with and without an SD.} \label{fig:CostVsTime}
\end{figure}

Having insight into the properties of the proposed Stackelberg game, we now show the performance of the proposed scheme when the SFC is equipped with an SD. For this purpose, we assume that the SFC has a 100~kWh SD with an efficiency of $0.9$ and a maximum charging-discharging rate of 24 kWh~\cite{Sousa-J-TSG:2012}. The price at different times of the day is obtained from~\cite{Jin-J-TVT:2013}, and the maximum and minimum price thresholds are considered to be $p_g^\text{min} = 40$ and $p_g^\text{max} = 45$ cents per kWh respectively. The demand of the SFC at different times of the day is chosen randomly from $\left[300, 700\right]$ kWh. For $5$ RUs in the system, we show the cost of the SFC at different times of the day for three different cases in Fig.~\ref{fig:CostVsTime}. These cases are 1) when the SFC does not have any SD and does not take part in the game,  2) when the SFC has an SD but does not take part in the game, and finally, 3) when the SFC has an SD and also plays the game with the RUs following Algorithm~\ref{alg:1}.

As can be seen from Fig.~\ref{fig:CostVsTime}, during the period when the grid price is low, the cost to the SFC is higher for cases 2 and 3 compared to case 1. In fact, due to the lower price, the SFC is more interested in charging its SD during this time so as to use it in peak hours. Hence, its required energy is more than the case without the SD. As a result, the cost is higher. However, the cost without the proposed game is considerably higher than the cost when the SFC and RUs interact with each other via Algorithm~\ref{alg:1}. The reason is that without playing the game, the SFC needs to buy all its energy from the grid including the energy for its SD. By contrast, the proposed game allows the SFC to pay the RUs a lower price than the grid's price to buy some of its required energy. Consequently, the SFC benefits in terms of its reduced total cost of energy purchase.

During peak hours, the cost to the SFC is significantly higher for case 1. In this case, the SFC needs to buy all its required energy from the grid at a significantly higher price. However, for the case when the SFC possesses an SD, the cost is lower as the stored energy allows the SFC to buy less from the grid compared to the previous case. Nevertheless, the most impressive performance is observed for case 3 when the SFC with an SD plays the Stackelberg game with the RUs following Algorithm~\ref{alg:1}. On the one hand, the stored energy  allows the SFC to buy a lower amount of energy during the peak hour like in case 2. On the other hand, unlike the other two cases, by taking part in the Stackelberg game the SFC manages to buy a certain fraction of its requirement from the RUs at a cheaper rate, compared to the grid's price, which minimizes its total cost of energy purchase noticeably. As Fig.~\ref{fig:CostVsTime} shows, on average, the cost reduction of the proposed case is $53.8\%$ compared to the case in which the SFC does not take part in the game. As can be seen from Fig.~\ref{fig:CostVsTime}, the performance is even  more impressive when compared to case 1.

\begin{figure}[t]
\centering
\includegraphics[width=\columnwidth]{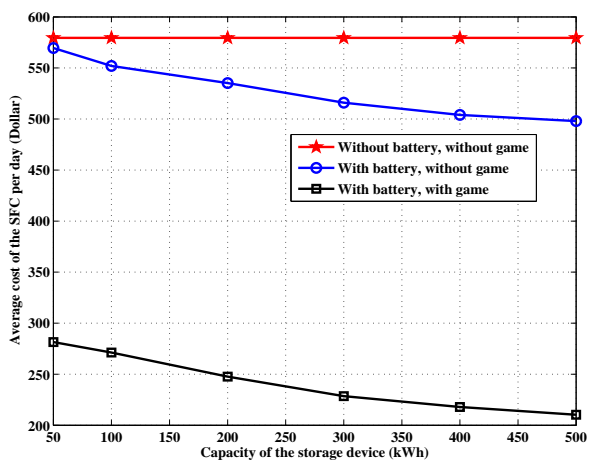}
\caption{Change in cost savings with the change of the capacity and charging rate of the SD.} \label{fig:CostSavings}
\end{figure}
However, cost savings are greatly affected by the SD's characteristics such as its capacity. In fact, for a particular charging rate, higher capacity can considerably assist the SFC to improve its performance in terms of average cost savings during a day, as shown in Fig.~\ref{fig:CostSavings}. As the figure illustrates, the cost saving for an SFC by using an SD increases with an increase in the capacity of the SD. According to Fig.~\ref{fig:CostSavings}, the savings in daily cost are, on average, $54.02\%$ and $58.1\%$ per day for the proposed case compared to the cases when the SFC does not play the game and when the SFC neither plays the game nor has any SDs respectively.

\section{Conclusion}\label{sec:conclusion}
In this paper, we have presented an energy management scheme for a smart community using a non-cooperative Stackelberg game. We have designed a system model suitable for applying the game, and have shown the existence of a strategy proof, unique Stackelberg equilibrium (SE), by exploring the properties of the game. We have shown that the use of distributed energy resources (DERs) is beneficial for both the shared facility controller (SFC) and residential units (RUs) at the SE. We have proposed a distributed algorithm, which is guaranteed to reach the SE of the game. Further, we have extended the scheme to the case in which the SFC has a storage device (SD). We have designed an effective charging-discharging scheme, for the SFC's SD based on the grid's price, which has been shown to have considerable influence on the cost incurred by the SFC. By the proposed charging-discharging scheme, the average cost to the SFC during a day has been shown to be reduced markedly compared to the case without a SD.

The proposed work can be extended in various ways. An interesting extension would be to check the impact of discriminate pricing among the RUs on the outcome of the scheme. Another compelling addition would be to determine how to set the threshold on the grid's price. Furthermore, quantifying the inconvenience that the SFC/RUs face during their interaction is another possible future investigation based on this work.

\begin{IEEEbiography}[{\includegraphics[width=1in,height=1.25in,clip,keepaspectratio]{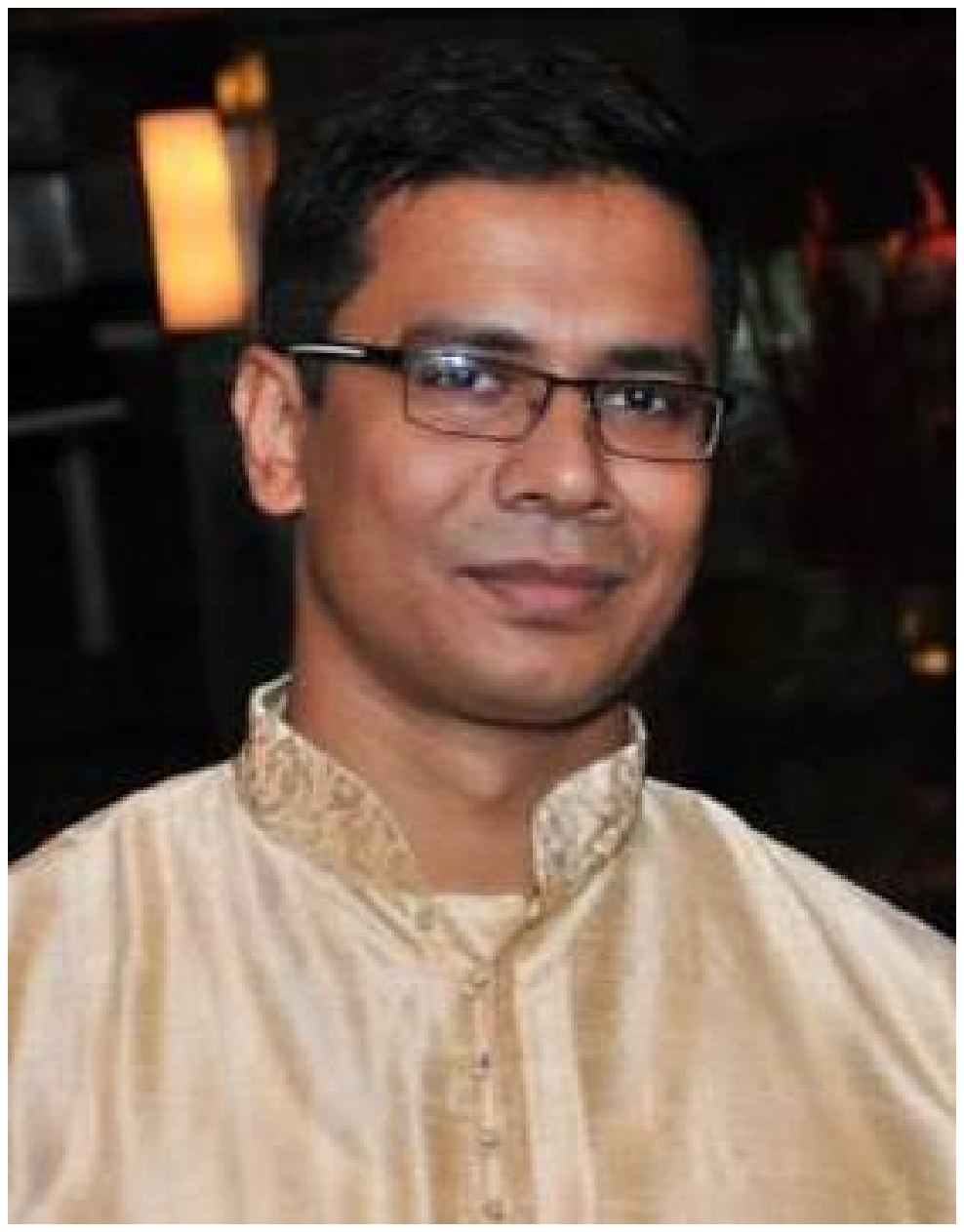}}]
{Wayes Tushar}(S'06, M'13)  received the B.Sc. degree in Electrical and Electronic Engineering from Bangladesh University of Engineering and Technology (BUET), Bangladesh, in 2007 and the Ph.D. degree in Engineering from the Australian National University (ANU), Australia in 2013. Currently, he is a postdoctoral research fellow at Singapore University of Technology and Design (SUTD), Singapore. Prior joining SUTD, he was a visiting researcher at National ICT Australia (NICTA) in ACT, Australia. He was also a visiting student research collaborator in the School of Engineering and Applied Science at Princeton University, NJ, USA during summer 2011. His research interest includes signal processing for distributed networks, game theory and energy management for smart grids. He is the recipient of two best paper awards, both as the first author, in Australian Communications Theory Workshop (AusCTW), 2012 and IEEE International Conference on Communications (ICC), 2013.
\end{IEEEbiography}
\begin{IEEEbiography}[{\includegraphics[width=1in,height=1.25in,clip,keepaspectratio]{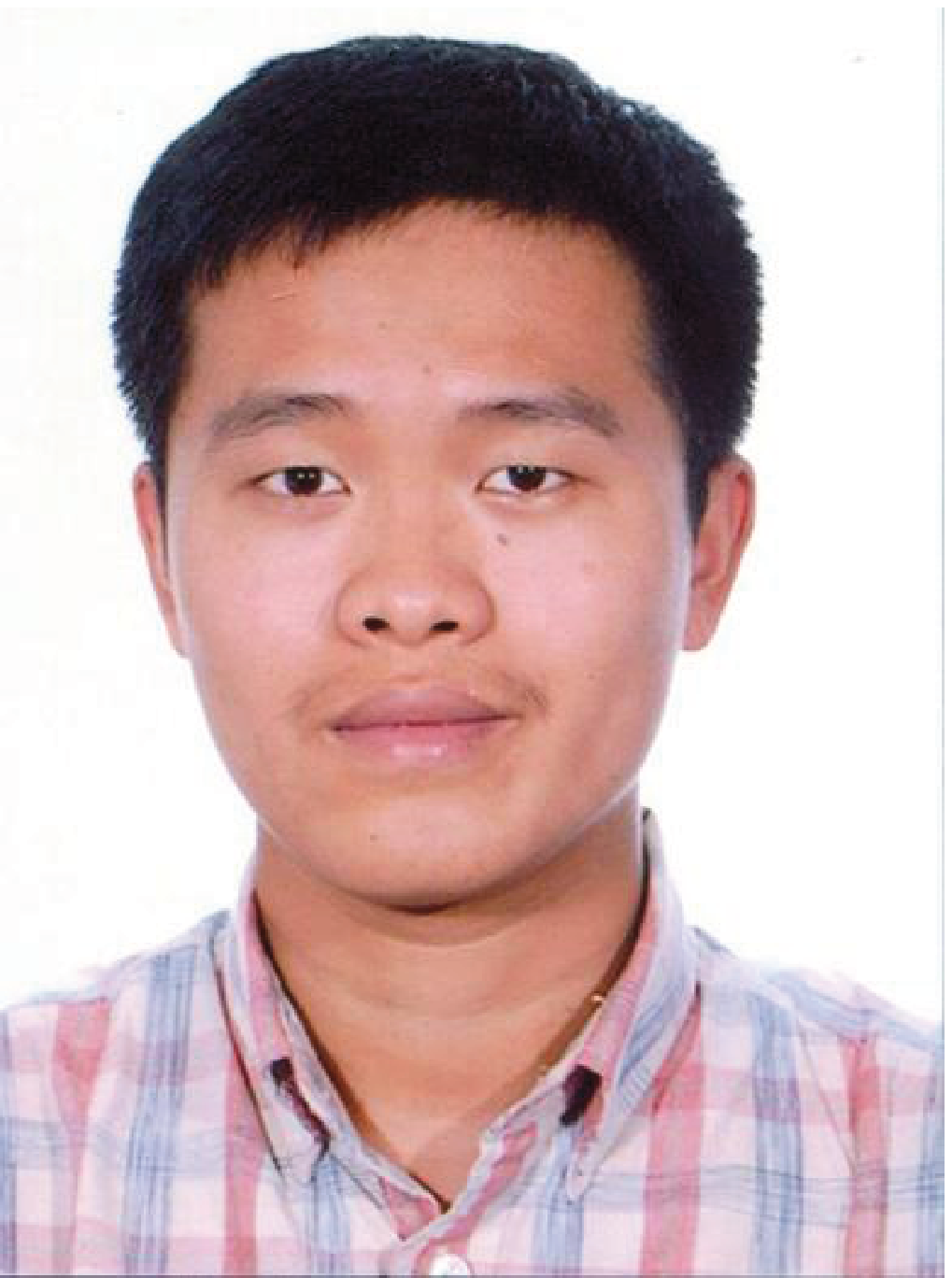}}]
{Bo Chai} received the B.Sc degree in Automation from Zhejiang University, Hangzhou, China, in 2010, when he also graduated from Chu Kochen 
Honors College. He was a visiting student at Simula Research Laboratory, Oslo, Norway in 2011. Currently, he is a member of the Group of Networked Sensing and Control (IIPC-nesC) in the State Key Laboratory of Industrial Control Technology, Zhejiang University, and he is also a visiting student in Singapore University of Technology and Design from 2013 to 2014. His research interests include smart grid and cognitive radio.
\end{IEEEbiography}
\begin{IEEEbiography}[{\includegraphics[width=1in,height=1.25in,clip,keepaspectratio]{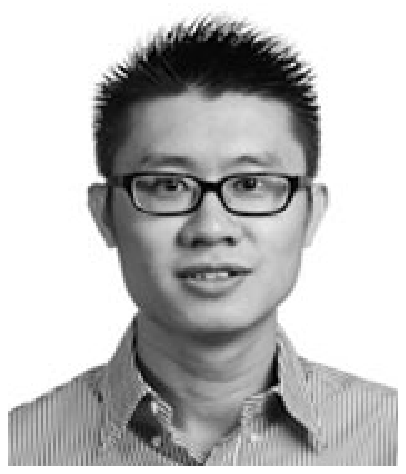}}]
{Chau Yuen} received the B. Eng and PhD degree from Nanyang Technological University, Singapore in 2000 and 2004 respectively. Dr Yuen was a Post Doc Fellow in Lucent Technologies Bell Labs, Murray Hill during 2005. He was a Visiting Assistant Professor of Hong Kong Polytechnic University in 2008. During the period of 2006 2010, he worked at the Institute for Infocomm Research (Singapore) as a Senior Research Engineer. He joined Singapore University of Technology and Design as an assistant professor from June 2010. He serves as an Associate Editor for IEEE Transactions on Vehicular Technology. On 2012, he received IEEE Asia-PaciÞc Outstanding Young Researcher Award.
\end{IEEEbiography}
\begin{IEEEbiography}[{\includegraphics[width=1in,height=1.25in,clip,keepaspectratio]{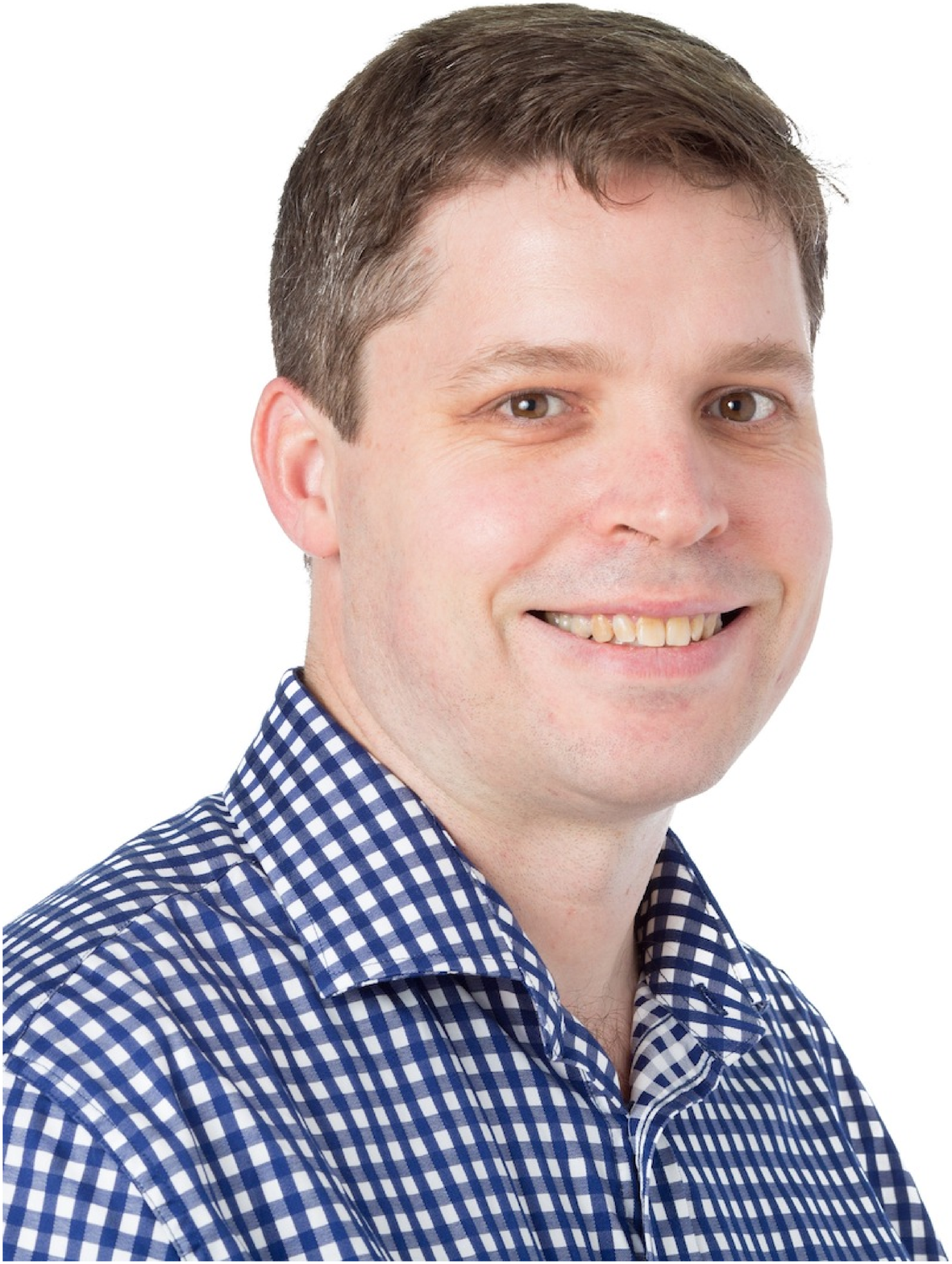}}]
{David Smith} is a Senior Researcher at National ICT Australia (NICTA) and is an adjunct Fellow with the Australian National University (ANU), and has been with NICTA and the ANU since 2004. He received the B.E. degree in Electrical Engineering from the University of N.S.W. Australia in 1997, and while studying toward this degree he was on a CO-OP scholarship. He obtained an M.E. (research) degree in 2001 and a Ph.D. in 2004 both from the University of Technology, Sydney (UTS), and both in Telecommunications Engineering. His research interests are in technology and systems for wireless body area networks; game theory for distributed networks; mesh networks; disaster tolerant networks; radio propagation and electromagnetic modeling; MIMO wireless systems; coherent and non-coherent space-time coding; and antenna design, including the design of smart antennas. He also has research interest in distributed optimization for smart grid. He has also had a variety of industry experience in electrical engineering; telecommunications planning; radio frequency, optoelectronic and electronic communications design and integration. He has published over 80 technical refereed papers and made various contributions to IEEE standardization activity; and has received four conference best paper awards.
\end{IEEEbiography}
\begin{IEEEbiography}[{\includegraphics[width=1in,height=1.25in,clip,keepaspectratio]{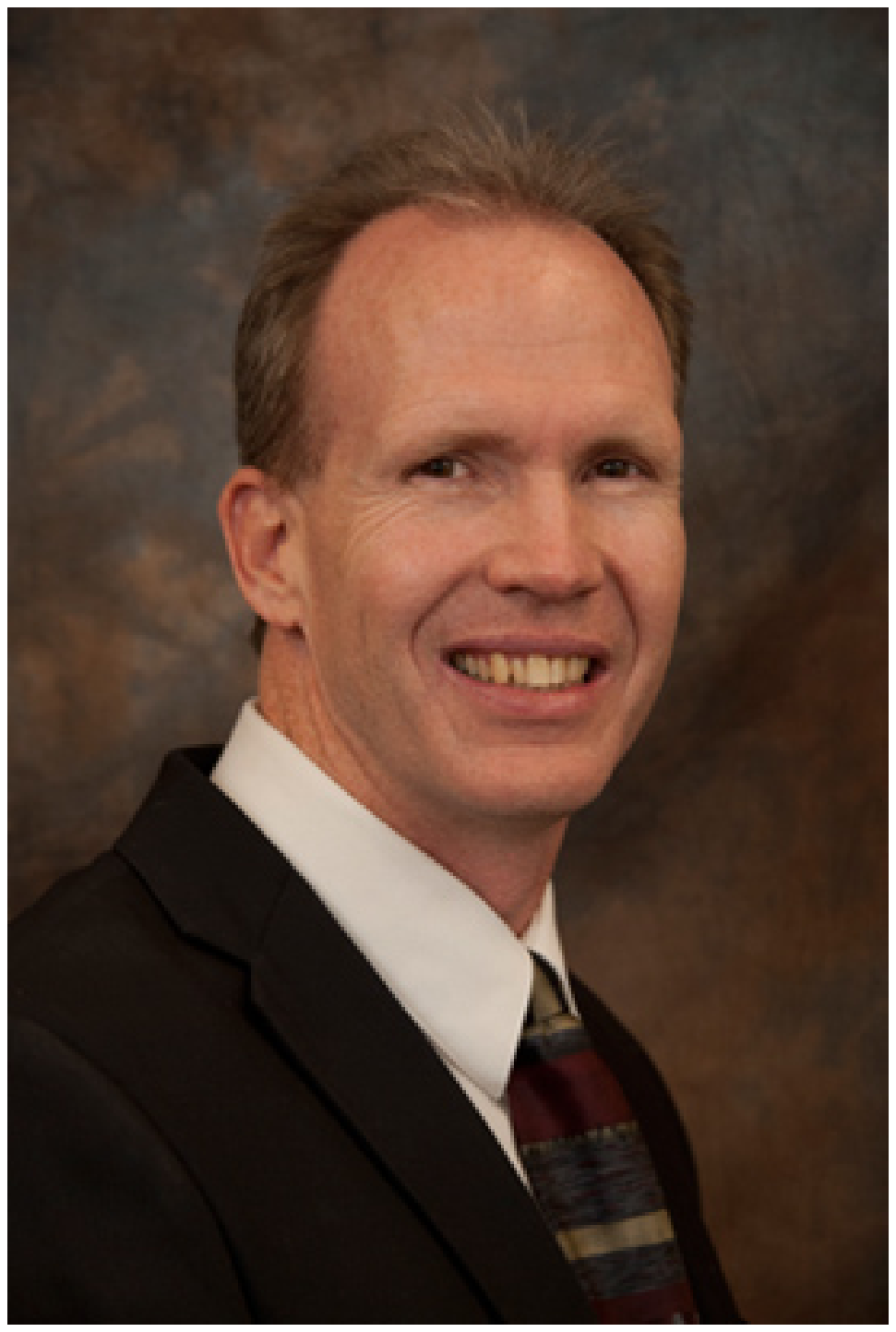}}]
{Kristin L. Wood} joined the faculty at the University of Texas in September 1989 after completing his doctoral work and established a computational and experimental laboratory for research in engineering design and manufacturing, in addition to a teaching laboratory for prototyping, reverse engineering measurements, and testing. During the 1997-98 academic year, Dr Wood was a Distinguished Visiting Professor at the United States Air Force Academy where he worked with USAFA faculty to create design curricula and research within the Engineering Mechanics / Mechanical Engineering Department. Through 2011, Dr Wood was a Professor of Mechanical Engineering, Design \& Manufacturing Division at The University of Texas at Austin. He was a National Science Foundation Young Investigator, the ``Cullen Trust for Higher Education Endowed Professor in Engineering", ``University Distinguished Teaching Professor", and the Director of the Manufacturing and Design Laboratory (MaDLab) and MORPH Laboratory.

Dr Wood has published more than 300 commentaries, refereed articles and books, and has received three ASME Best Research Paper Awards, two ASEE Best Paper Awards, an ICED Best Research Paper Award, the Keck Foundation Award for Excellence in Engineering Education, the ASEE Fred Merryfield Design Award, the NSPE AT\&T Award for Excellence in Engineering Education, the ASME Curriculum Innovation Award, the Engineering Foundation Faculty Excellence Award, the Lockheed Martin Teaching Excellence Award, the Maxine and Jack Zarrow Teaching Innovation Award, the Academy of University Distinguished Teaching Professors' Award, and the Regents' Outstanding Teacher Award. 
\end{IEEEbiography}
\begin{IEEEbiography}[{\includegraphics[width=1in,height=1.25in,clip,keepaspectratio]{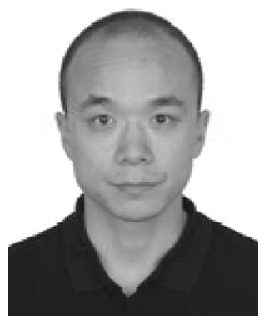}}]
{Zaiyue Yang} (M'10) received the B.S. and M.S. degrees from Department of Automation, University of Science and Technology of China, Hefei, China, in 2001 and 2004, respectively, and the Ph.D. degree from Department of Mechanical Engineering, University of Hong Kong, Hong Kong, in 2008. He then worked as Postdoctoral Fellow and Research Associate in the Department of Applied Mathematics, Hong Kong Polytechnic University before joining Zhejiang University, Hangzhou, China, in 2010. He is currently an associate professor there. His current research interests include smart grid, signal processing and control theory.
\end{IEEEbiography}
\begin{IEEEbiography}[{\includegraphics[width=1in,height=1.25in,clip,keepaspectratio]{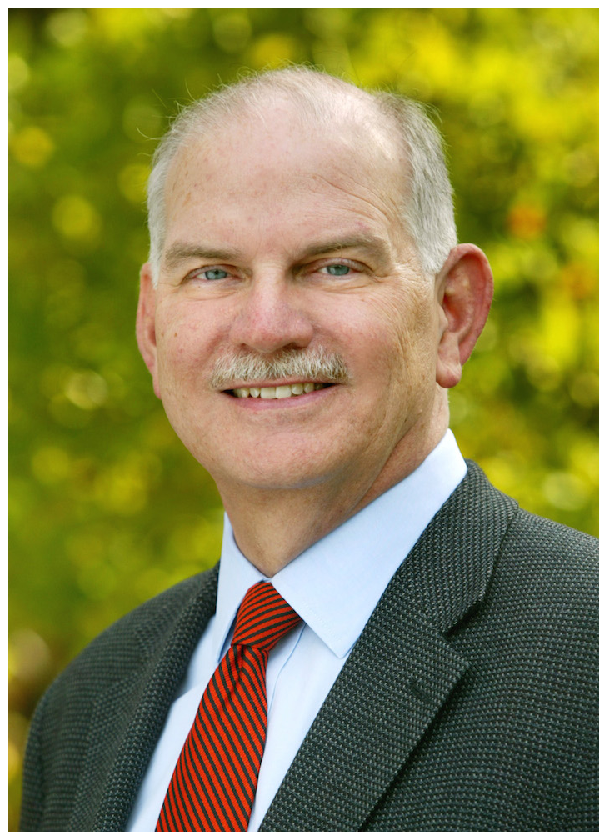}}]
{H. Vincent Poor} (S'72, M'77, SM'82, F'87) received the Ph.D. degree in EECS from Princeton University in 1977.  From 1977 until 1990, he was on the faculty of the University of Illinois at Urbana-Champaign. Since 1990 he has been on the faculty at Princeton, where he is the Michael Henry Strater University Professor of Electrical Engineering and Dean of the School of Engineering and Applied Science. Dr. Poor?s research interests are in the areas of stochastic analysis, statistical signal processing, and information theory, and their applications in wireless networks and related fields such as social networks and smart grid. Among his publications in these areas are the recent books \emph{Principles of Cognitive Radio} (Cambridge University Press, 2013) and \emph{Mechanisms and Games for Dynamic Spectrum Allocation} (Cambridge University Press, 2014).

Dr. Poor is a member of the National Academy of Engineering and the National Academy of Sciences, and is a foreign member of Academia Europaea and the Royal Society.   He is also a  fellow of the American Academy of Arts \& Sciences, the Royal Academy of Engineering (UK) and the Royal Society of Edinburgh. In 1990, he served as President of the IEEE Information Theory Society, and in 2004-07 he served as the Editor-in-Chief of the \emph{IEEE Transactions on Information Theory}. He received a Guggenheim Fellowship in 2002 and the IEEE Education Medal in 2005. Recent recognition of his work includes the 2014 URSI Booker Gold Medal, and honorary doctorates from Aalborg University, the Hong Kong University of Science and Technology and the University of Edinburgh.  
\end{IEEEbiography}
\end{document}